\documentclass[11pt]{article}

\usepackage{amsmath,amsthm,amssymb}
\usepackage{hyperref}
\usepackage[margin=1in]{geometry}
\usepackage{graphicx}
\usepackage{enumerate}
\usepackage{caption}
\usepackage{subcaption}

\theoremstyle{definition}
\newtheorem{thm}{Theorem}
\newtheorem{hyp}[thm]{Hypothesis}
\newtheorem{cor}[thm]{Corollary}

\numberwithin{equation}{section}
\numberwithin{thm}{section}

\newcommand{\wone}{w_1}
\newcommand{\wtwo}{w_2}
\newcommand{\w}{W}

\renewcommand{\a}{a}
\renewcommand{\b}{b}
\newcommand{\m}{m}
\newcommand{\mtwo}{1}

\newcommand{\blank}{\,\cdot\,}

\newcommand{\what}{\tilde{\wone}}
\newcommand{\whattwo}{\tilde{\wtwo}}
\newcommand{\whatgen}{\tilde{w}}

\renewcommand{\d}{\,\mathrm{d}}
\newcommand{\D}[1]{\frac{\d #1}{\d x}}

\newcommand{\hypref}[1]{{(H{\ref{#1}})}}

\usepackage{authblk}

\title{Periodic Traveling Waves in an Integro-Difference Equation With Non-Monotonic Growth and Strong Allee Effect}
\date{}
\author{Michael Nestor and Bingtuan Li\thanks{M. Nestor's email is \href{mailto:mdnest01@gmail.com}{mdnest01@gmail.com}. B. Li was partially supported by the National Science Foundation under Grant DMS-1515875 and Grant DMS-1951482.}}

\affil{Department of Mathematics, University of Louisville\linebreak Louisville, KY 40292}

\begin{document}

\maketitle

\begin{abstract}
We derive sufficient conditions for the existence of a periodic traveling wave solution to an integro-difference equation with a piecewise constant growth function exhibiting a stable period-2 cycle and strong Allee effect.
The mean traveling wave speed is shown to be the asymptotic spreading speed of solutions with compactly supported initial data under appropriate conditions.
We then conduct case studies for the Laplace kernel and uniform kernel.
\end{abstract}

\newline


\section{Introduction}

Integro-difference equations in the form 
\begin{align}\label{q}
u_{n+1}(x)\;=\;Q[u_n](x)\;:=\,\int^{\infty}_{-\infty}k(x-y)\,g\left(u_n(y)\right) \d y
\end{align}
are of great interest in the studies of invasions of populations with discrete generations and separate growth and dispersal stages.
They have been used to predict changes in gene frequency \cite{lui82a, lui82b, lui83, slatkin, w78}, and applied to ecological problems~\cite{hh, ks, kot89, kot92, kotbook, lut, nkl,otto}.
Previous rigorous studies on integro-difference equations have assumed that the growth function is nondecreasing~\cite{w78, wein82}, or is non-monotone without Allee effect~\cite{lui83, wang}.
The results show existence of constant spreading speeds and travelng waves with fixed shapes and speeds.
Sullivan et al. ~\cite{pnas} demonstrated numerically that an integro-difference equation with a non-monotone growth function exhibiting a strong Allee effect can generate traveling waves with fluctuating speeds, and Otto~\cite{otto} showed that such an equation can have robust non-spreading solutions.
In this paper we give conditions for the existence of a periodic traveling wave with two intermediate speeds for an integro-difference equation with a piecewise constant growth function that exhibits a strong Allee effect and a period-2 cycle. 

Piecewise constant growth functions have been used in the studies of integro-difference equations; see for example~\cite{kot1, lut, otto, pnas}.
Such an equation is analytically tractable and it can provide specific insights into the dynamics of solutions.
For the piecewise constant growth function
\begin{equation} \label{g0}
g(u) = \begin{cases}
0, & \text{if } u < \a, \\
1, & \text{if } u \geq \a,
\end{cases}
\end{equation}
with $0<a<1$, which exhibits a strong Allee effect and monotonicity, Kot et al.~\cite{kot1} and Lutscher~\cite{lut} (Section 6.1) investigated the range expansion of the population and the spreading speed, and Sullivan et al.
~\cite{pnas} studied oscillations in spreading speeds when the dispersal kernel is density dependent.
More discussions about integro-difference equations with piecewise growth functions can be found in Lutscher~\cite{lut} (Chapter 15). 

In this paper, we consider the integro-difference equation (\ref{q}) with 
\begin{equation} \label{g}
g(u) = \begin{cases}
0, & \text{if } u < \a, \\
\mtwo, & \text{if } \a \leq u \leq \b, \\
\m, & \text{if } u > \b,
\end{cases}
\end{equation}
with $0<\a<\m<\b<1$.
$g(u)$ is a piecewise constant non-monotone growth function exhibiting a strong Allee effect~\cite{all}.
Specifically, it has a stable fixed point at zero and a stable period two cycle $(1,\m)$ with $\a$ the Allee threshold value.
This is the function considered in Otto~\cite{otto} where non-spreading solitions are studied.
It may be viewed as an extension of (\ref{g0}).
A graphical demonstration of $g$ defined by (\ref{g}) is given by Figure 1.
\begin{figure}[h!]
\centering
 \caption{The growth of $g(u)$ defined by (\ref{g}) with $0<\a<\m<\b<1$.}
 \includegraphics[width=.4\linewidth]{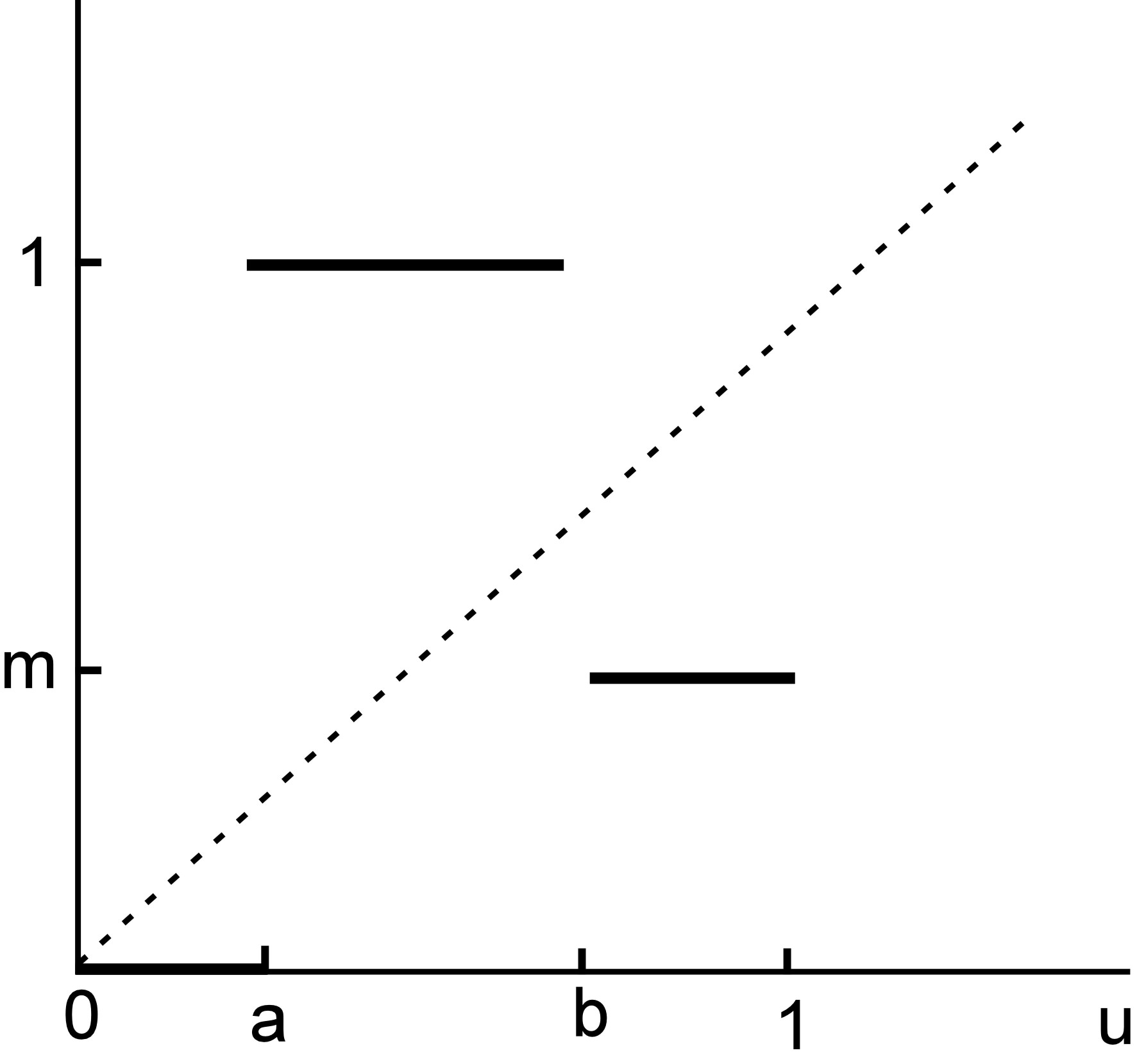}
\end{figure}
We rigorously construct periodic traveling waves with two intermediate speed for \eqref{q} with $g$ given by \eqref{g} and a proper dispersal kernel $k$.
To the best of our knowledge, this is the first time that traveling waves with oscillating speeds have been analytically established for scalar spatiotemporal equations with constant parameters.
We also show that the mean traveling wave speed is the spreading speed of solutions with compactly supported initial data under appropriate conditions.
We finally conduct case studies for the Laplace kernel and uniform kernel.

\section{Results}

In this section, we present our main results about our model.
Throughout, we assume all population density functions are bounded elements of $C(\mathbb R)$, the space of real-valued continuous functions.
It is straightforward to verify this space is closed under iteration of $Q$.
We denote the norm of a function $u \in C(\mathbb R)$ by $||u||_\infty$.

We make the following assumptions about the dispersal kernel.
\begin{hyp}
$k(x)$ is a piecewise differentiable function satisfying
\begin{enumerate}[(H1)]
\item $k(x) \geq 0$ for all $x$, and $\int_{-\infty}^{\infty} k(x) \d x = 1$, \label{h1}

\item $k(x)=k(-x)$ for all $x$,\label{h2}

\item $k(x)>0$ for all $x \in (-\sigma,\sigma)$ for some $0 < \sigma \leq \infty$,\label{h3}

\item \label{h4}
for all $\lambda \in (0,1)$, for all $A,B\in\mathbb R$ with $A>B$, the expression
$$
k(x-A) - \lambda k(x-B)
$$
has $<2$ sign changes, and

\item \label{h5}
for all $\lambda \in (0,1)$, for all $A,B\in\mathbb R$ with $A>B$, the expression
$$
k(x-r-A) - \lambda k(x-r-B) - k(x+r+A) + \lambda k(x+r+B)
$$
has $<4$ sign changes for sufficiently large $r>0$.
\end{enumerate}
\end{hyp}

\subsection{Periodic traveling waves}

Let $\wone,\wtwo \in C(\mathbb R)$ be bounded, continuous functions defined by
\begin{equation} \label{w1}
\wone(x) :=\int_x^\infty k(y) \d y
\end{equation}
and
\begin{equation} \label{w2}
\wtwo(x) := Q[\wone](x+c^*) = \int_{-\infty}^{\infty} k(x+c^*-y) g(\wone(y)) \d y
\end{equation}
where $c^*$ is a constant defined by
\begin{equation} \label{c}
c^* := \frac{1}{2} \sup \{ x\in\mathbb R \mid Q[\wone](x) = \a \}.
\end{equation}

The following theorem shows that $\wone$ and $\wtwo$ generate a traveling wave-like solution to \eqref{q}. We term this solution a {\it periodic traveling wave} since it alternates between two distinct wave profiles (with different limits at $-\infty$) while traveling with mean speed $c^*$.

\begin{thm} \label{theorem1} 
Let $(\w_n)_{n=0}^{\infty}$ be a sequence of bounded, continuous functions defined by
\begin{equation}\label{ptw}
\w_n(x) := \begin{cases}
\wone(x-nc^*), & n \text{ even}, \\
\wtwo(x-nc^*), & n \text{ odd}.
\end{cases}
\end{equation}
If $||w_2||_\infty < \b$, then $\w_{n+1}=Q[\w_n]$ for all $n\geq0$.
\end{thm}

\begin{proof}
For $i = 1,2$, define the right inverse of $w_i$ by
\begin{equation}
\Phi_i(\ell) := \sup \{ x \in \mathbb R \mid w_i(x) = \ell \}, \quad i=1,2.
\end{equation}
Note that $\Phi_i(\ell)$ is finite for $0<\ell<||w_i||_\infty$, where $||w_1||_\infty=1$, and $\a<||w_2||_\infty<\b$ by the assumption of this theorem.

By the translation invariance property of $Q$, it suffices to show
\begin{equation} \label{qw1}
Q[w_1](x) = w_2(x-c^*)
\end{equation}
and
\begin{equation} \label{qw2}
Q[w_2](x) = w_1(x-c^*)
\end{equation}
for all $x \in \mathbb R$.
The first equation follows immediately from definition \eqref{w2}.
To prove the second, observe that $w_1$ is monotonically decreasing and satisfies $w_1(\infty)=0$ and $w_1(-\infty)=1$ (see for example Figure \ref{fig:fig2}).
Thus,
$$
g(w_1(x)) = \begin{cases}
\m, & x < \Phi_1(\b), \\
1, & \Phi_1(\b) \leq x \leq \Phi_1(\a), \\
0, & x > \Phi_1(\a).
\end{cases}
$$
Applying the integro-difference operator yields
$$ \begin{aligned}
w_2(x-c^*) &= Q[w_1](x) \\
&= \m \int_{-\infty}^{\Phi_1(\b)} k(x-y) \d y + \int_{\Phi_1(\b)}^{\Phi_1(\a)} k(x-y) \d y \\
&= w_1(x-\Phi_1(\a)) - (1-\m)\, w_1(x-\Phi_1(\b)).
\end{aligned}$$
We can then differentiate with respect to $x$:
$$
\D{w_2} \Big| _{x-c^*}= - k(x-\Phi_1(\a)) + (\mtwo-\m) \, k(x-\Phi_1(\b)).
$$
From \hypref{h4}, $\D{w_2}$ has at most 1 sign change; thus, $w_2$ has at most 1 turning point.
We also have $w_2(x)<b$ for all $x$, hence $g(w_2(x)) \in \{0, 1 \}$.

We claim
\begin{equation} \label{gw2}
g(w_2(x)) = \begin{cases}
1, & x \leq c^*, \\
0, & x > c^*,
\end{cases}
\end{equation}
where $c^*$ is defined by \eqref{c}.
This can be shown by arguing in cases depending on the number of turning points of $w_2$.
\begin{enumerate}[{Case} 1.]
\item If $w_2$ has no turning points, then $w_2$ is monotone decreasing with $0 \leq w_2(x) \leq \m$ for all $x$.
It follows by the intermediate value theorem (IVT) that $c^* = \Phi_2(\a)$ satisfies \eqref{gw2}.

\item If $w_2$ has one turning point, say $x_0 \in \mathbb R$, then $w_2$ must be increasing for $x<x_0$ and decreasing for $x>x_0$.
It follows that $x_0$ is a global maximum, and $w_2$ is increasing on $(-\infty,x_0)$ and decreasing on $(x_0,\infty)$.
We have $g(w_2(x))=1$ for $x<x_0$.
By the IVT, $c^* = \Phi_2(\a)$ is the unique solution to \eqref{gw2}.
\end{enumerate}
Taking the convolution of equation \eqref{gw2} with $k(x)$ yields
$$
Q[w_2](x) \,=\, \int_{-\infty}^{c^*} k(x-y) \d y \,=\, w_1(x-c^*).
$$
This completes the proof.
\end{proof}

\begin{figure}[h]
\begin{subfigure}{0.5\textwidth}
\includegraphics[width=\linewidth]{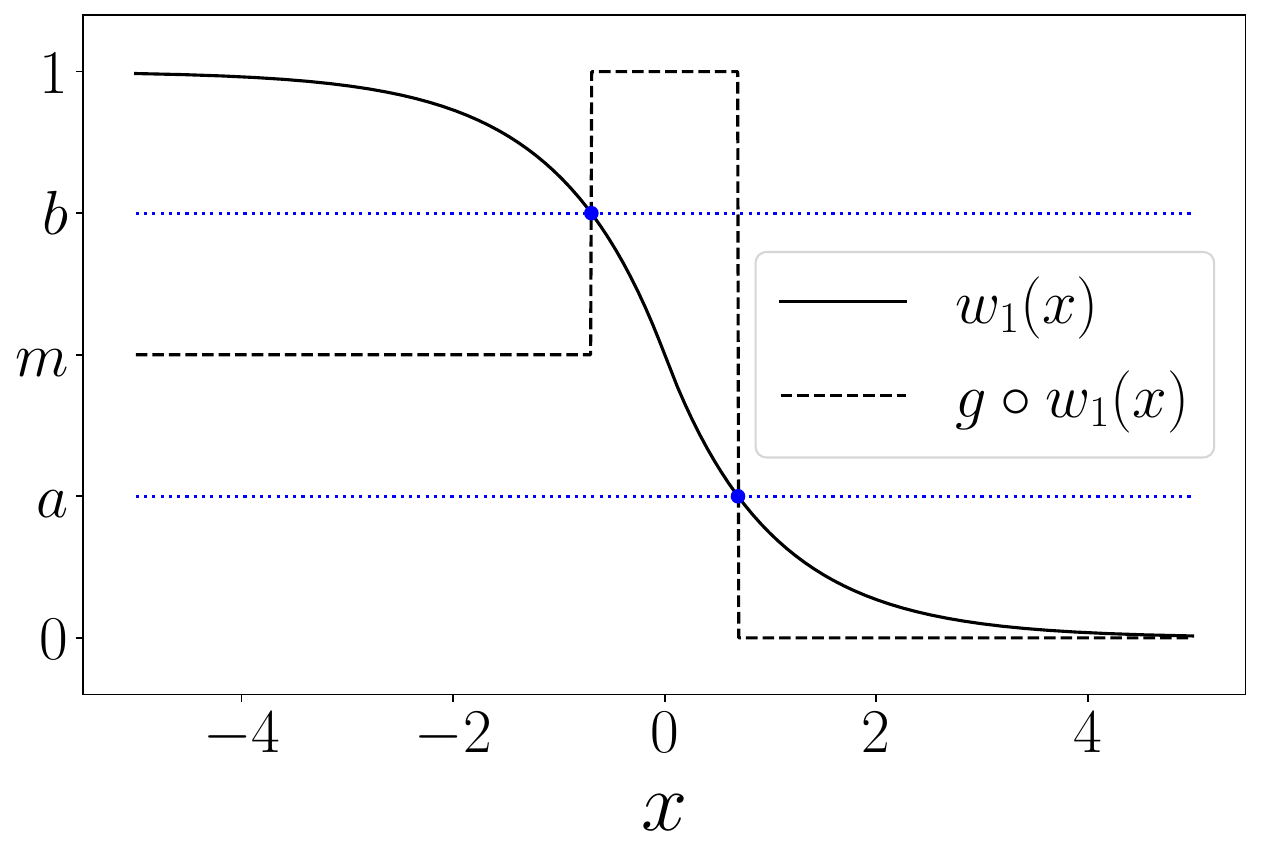}
\caption{$\wone(x)$}
\end{subfigure}
\begin{subfigure}{0.5\textwidth}
\includegraphics[width=\linewidth]{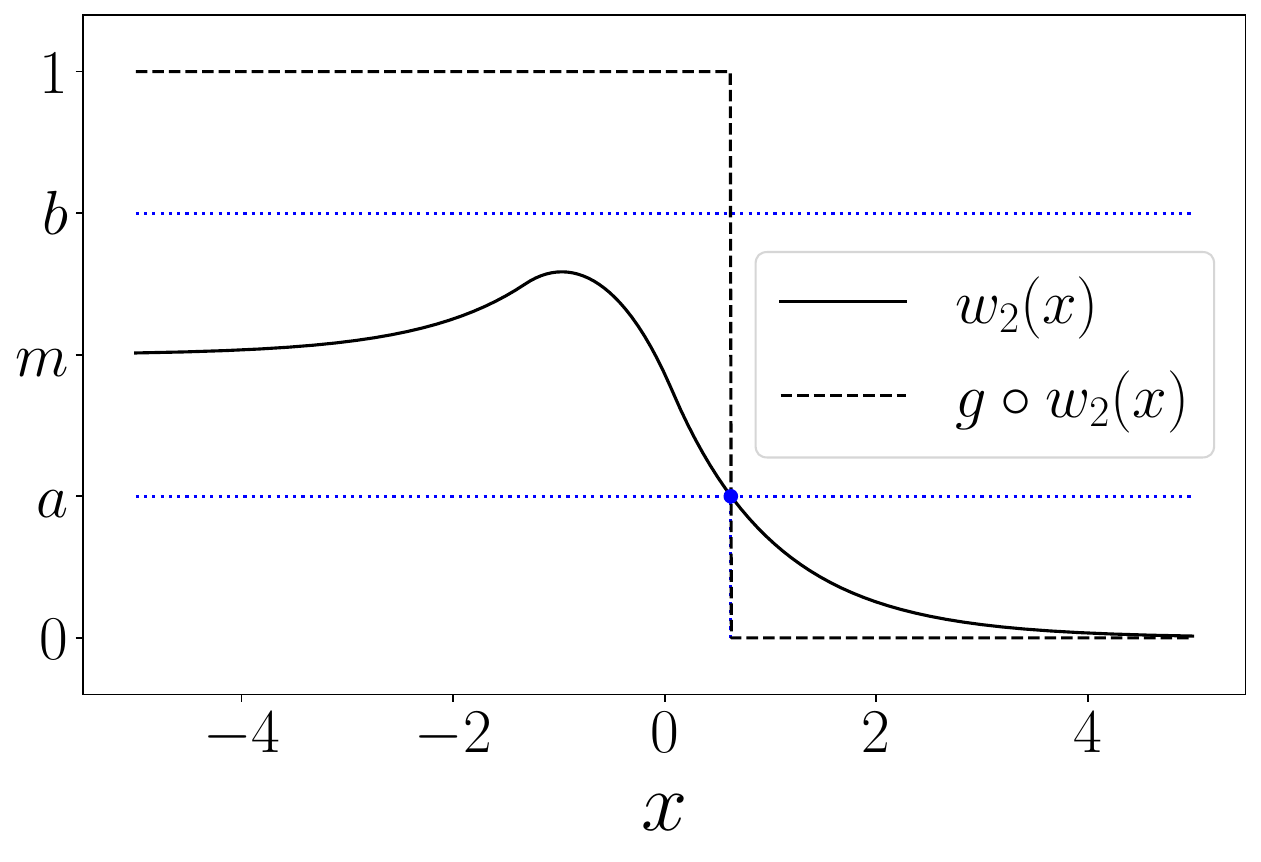}
\caption{$\wtwo(x)$}
\end{subfigure}
\caption{Profiles of the intermediate traveling waves $w_1$ and $w_2$, with $k(x)$ equal to the Laplace kernel (see Section 3).
Blue points mark the intersection of the curves with the level sets $\ell=\a$ and $\ell=\b$.}
\label{fig:fig2}
\end{figure}

Theorem \ref{theorem1} describes a rightward periodic traveling wave, but we can define the corresponding leftward periodic traveling wave via $x\mapsto\w_n(-x)$.
This is due to the fact that $Q$ is symmetric and translation invariant.

We investigated the intermediate wave speed generated by the periodic traveling wave solution \eqref{ptw}.
To do this, we define a sequence $(\xi_{n})_{n=0}^{\infty}$ of continuous functions $(0,\m)\to\mathbb R$ by

\begin{equation}
 \xi_{n}(\ell) := \sup \{ x \in \mathbb R : \w_n(x) = \ell \}, \quad \ell \in (0,\m).
\end{equation}
This expression yields the rightmost intersection of the curve $u=\w_n(x)$ with the horizontal line $u=\ell$.
By computing the first difference, we obtain a new function sequence $(c_n)_{n=0}^{\infty}$ (with the same domain and range) defined by
\begin{equation}
c_n(\ell) := \xi_{n+1}(\ell) - \xi_n(\ell), \quad \ell \in (0,\m).
\end{equation}
From Theorem \ref{theorem1}, we have that $c_n(\ell)$ converges to a period-$2$ cycle for every $\ell \in (0,\m)$.

Our results about the intermediate wave speed sequence are summarized in the following corollary.

\begin{cor}\label{cor1}
Assume the hypothesis of Theorem \ref{theorem1} holds. Then for all $\ell \in (0, \m)$, the sequence $c_n(\ell)$ satisfies
\begin{equation}
c_n(\ell) = c^* + \begin{cases}
\Phi_1(\ell) - \Phi_2(\ell), & n \text{ even}, \\
\Phi_2(\ell) - \Phi_1(\ell), & n \text{ odd}.
\end{cases}
\end{equation}
\end{cor}

\begin{proof}
By Theorem \ref{theorem1}, we have
$$ \begin{aligned}
\xi_n(\ell) &= \begin{cases}
\sup\{x \in \mathbb R \mid w_1(x-nc^*) = \ell \}, & n \text{ even},\\
\sup\{x \in \mathbb R \mid w_2(x-nc^*) = \ell \}, & n \text{ odd},
\end{cases} \\
&= nc^* + \begin{cases}
\Phi_1(\ell), & n \text{ even}, \\
\Phi_2(\ell), & n \text{ odd}.
\end{cases}
\end{aligned}
$$
for every $\ell \in (0,\m)$. Taking the first difference yields
$$ \begin{aligned}
c_n(\ell) &= \xi_n(\ell) - \xi_{n-1}(\ell) \\
&= nc^* + \begin{cases}
\Phi_1(\ell), & n \text{ even}, \\
\Phi_2(\ell), & n \text{ odd},
\end{cases} - (n-1)\, c^* + \begin{cases}
\Phi_2(\ell), & n \text{ even}, \\
\Phi_1(\ell), & n \text{ odd},
\end{cases} \\
&= c^* + \begin{cases}
\Phi_1(\ell) - \Phi_2(\ell), & n \text{ even}, \\
\Phi_2(\ell) - \Phi_1(\ell), & n \text{ odd}.
\end{cases}
\end{aligned} $$
\end{proof}

Thus, for each fixed $\ell \in (0,\m)$, we can define intermediate traveling wave speeds $c_1^*$ and $c_2^*$ by
\begin{equation} \label{c1star}
c_1^*(\ell) := c^* + \Phi_1(\ell) - \Phi_2(\ell)
\end{equation}
and
\begin{equation} \label{c2star}
c_2^*(\ell) := c^* + \Phi_2(\ell) - \Phi_1(\ell).
\end{equation}

It follows that $\eqref{q}$ has a periodic traveling wave with wave profiles $w_1(x)$ and $w_2(x)$, intermediate wave speeds $c_1^*(\ell)$ and $c_2^*(\ell)$, and mean wave speed $\frac{1}{2}\left(c_1^*(\ell) + c_2^*(\ell)\right)=c^*$.
Note that the intermediate wave speeds depend on $\ell$, but their sum is identically equal to $2c^*$.

\subsection{Periodic spreading solutions}

By applying the results of Theorem \ref{theorem1}, we were able to prove the asymptotic spreading speed of solutions with compact initial data.
In particular, we showed that $c^*$ is the asymptotic spreading speed of compactly supported initial data with sufficient weight above the Allee threshold but below the overcompensation threshold.

\begin{thm} \label{theorem2}
Assume the hypothesis of Theorem \ref{theorem1} holds. Suppose $u \in C(\mathbb R)$ is bounded, non-negative, and has compact support. If
\begin{enumerate}[i.]
\item $||u|| _\infty < \b$,
\item the set $A := \{ x \in \mathbb R : u(x) \geq \a \} $ is connected and has sufficiently large length, and
\item $c^* > 0$,
\end{enumerate}
then the sequence $(u_n)_{n=0}^{\infty}$ defined by $u_0=u$ and $u_{n+1}=Q[u_n]$ for $n\geq0$ spreads asymptotically with mean speed $c^*$.
\end{thm}

\begin{proof}
We define two continuous functions $\what,\whattwo : \mathbb R \times [0,\infty)\to\mathbb R$ by
\begin{equation}
\what(x,r) := w_1(x) - w_1(x+2r) = \int_{-2r}^{0} k(x-y) \d y
\end{equation}
and
\begin{equation}
\whattwo(x,r) := Q[\what(\blank,r)](x+c^*).
\end{equation}
Observe that $\what$ is symmetric with respect to $x=-r$, and $\whattwo$ is symmetric with respect to $x=-r-c^*$.
This notation can be justified by observing that $\what(x,r) \to w_1(x)$ and $\whattwo(x,r)\to w_2(x)$ as $r\to\infty$ for all $x$.
We also set
\begin{equation}
\phi_i(r,\ell) := \sup \{ x \in \mathbb R \mid \whatgen_i(x,r) = \ell \}, \quad i=1,2.
\end{equation}
$\phi_1$ and $\phi_2$ are the right-inverse of $\what$ and $\whattwo$, respectively.
For a fixed $r>0$, the value of $\phi_i(r,\ell)$ is finite if $0 < \ell < || \whatgen_i(\blank, r) || _\infty$, with the latter bound converging to $||w_i||_\infty$ as $r\to\infty$.
They satisfy
$$
\lim_{r\to\infty} \phi_i (r,\ell) = \Phi_i(\ell) , \quad 0 < \ell < || w_i ||_\infty, \quad i=1,2.
$$

We now apply the growth function to $\what$, assuming $r$ is sufficiently large enough so that $\a < ||\what(\blank,r)||_\infty < \b$; this is guaranteed by taking $r > \Phi_1\left(\frac{1-a}{2}\right)$.
We have
$$
g\left(\what(x,r)\right) = \begin{cases}
0, & x < -2r - \phi_1(r,a), \\
1, & -2r - \phi_1(r,a) < x < -2r - \phi_1(r,b), \\
\m, & -2r - \phi_1(r,b) < x < \phi_1(r,b),\\
1, & \phi_1(r,b) < x < \phi_1(r,a), \\
0, & \phi_1(r,a) < x.
\end{cases} $$
Note that the above expression converges to $g(w_1(x))$ as $r\to\infty$ for all $x$.
We then apply the convolution operator:
$$\begin{aligned} 
\whattwo(x-c^*,r) = Q[\what(\blank,r)](x) &= \int_{-2r-\phi_1(r,a)}^{\phi_1(r,a)} k(x-y) \d y - (1-\m)\, \int_{-2r-\phi_1(r,b)}^{\phi_1(r,b)} k(x-y) \d y \\
&= w_1(x-\phi_1(r,a)) - (1-\m)\, w_1(x-\phi_1(r,b)) \\
&+ (1-\m)\, w_1(x+2r+\phi_1(r,b)) - w_1(x+2r+\phi_1(r,a)).
\end{aligned}$$
Computing the derivative of this expression yields
$$\begin{aligned} \label{qhatcomputation2}
\D{\whattwo}\Big|_{x-c^*} &= -k(x-\phi_1(r,a)) + (1-\m)\, k(x-\phi_1(r,b)) \\
& - (1-\m)\, k(x+2r+\phi_1(r,b)) + k(x+2r+\phi_1(r,a)).
\end{aligned}$$

By \hypref{h5}, the number of turning points of $\whattwo$ is at most 3 (assuming $r$ is sufficiently large).
Furthermore, $\whattwo$ is symmetric around $x=-r-c^*$; thus the sign changes must come in pairs or occur exactly at $x=-r-c^*$.
We can immediately exclude the cases of 0 and 2 turning points because $\whattwo$ is non-negative, vanishes at $\pm\infty$, and is not identically zero.
As in the proof of Theorem \ref{theorem1}, this leaves two cases:
\begin{enumerate}[{Case} 1.]

\item If $\whattwo$ has 1 turning point, then it must occur at $x=-r-c^*$.
We have $\whattwo(-r-c^*,r) = w_2(-r-c^*) - w_2(r+c^*) \to \m$ as $r \to \infty$; hence, it is sufficient to take $r$ large enough so that $|\whattwo(-r-c^*,r) - \m| < \min\{|\a-\m|,|\b-\m|\}$.
It follows by the triangle inequality that $\whattwo(x,r) \leq \whattwo(-r-c^*,r) < \b$ for all $x$.

\item If $\whattwo$ has 3 turning points, then again by symmetry, they must be given by $\{-r-c^*-t, -r-c^*, -r-c^*+t\}$, for some $t>0$.
We can label them in increasing order by $t_1<t_2<t_3$.
It follows that $t_1$ and $t_3$ are maxima with $\whattwo(t_1,r)=\whattwo(t_3,r) \to ||w_2||_\infty$ as $r\to\infty$, and $t_2$ is a local minima with $\whattwo(t_2,r) \to \m$ as $r\to\infty$.
Thus $\whattwo$ is monotone increasing on $(-\infty,t_1)\cup(t_2,t_3)$ and monotone decreasing on $(t_1,t_2)\cup(t_3,\infty)$.
By taking $r$ sufficiently large, we can make $\left| \whattwo(t_1,r)-||w_2||_\infty \right| = \left| \whattwo(t_3,r)-||w_2||_\infty \right| < \left| \b-||w_2||_\infty \right| $ and $\left| \whattwo(t_2,r)-\m \right| < \left| \a-\m \right|$.
\end{enumerate}

In both cases, we have $\a \leq \whattwo(x,r) \leq \b$ on a closed interval of radius $r + c^* + \phi_2(r,a)$ and $\whattwo(x,r) < a$ elsewhere.
Taking composition with $g$ yields
$$
g(\whattwo(x-c^*,r)) = \begin{cases}
0, & x < -2r - c^* -\phi_2(r,\a), \\
1, & -2r - c^* -\phi_2(r,\a) \leq x \leq c^* + \phi_2(r,\a) ,\\
0, & x > c^* + \phi_2(r,\a).
\end{cases}
$$
Hence,
$$ \begin{aligned}
Q^2 [\what(\blank,r)] (x) &= w_1(x-c^* - \phi_2(r,a)) - w_1(x+2r+c^*+\phi_2(r,a)) .
\end{aligned} $$
Shifting the expression left by $c^*+\phi_2(r,a)$ units yields
\begin{equation} \begin{aligned} \label{appliedqtwice}
Q^2 [\what(\blank,r)] (x+c^*+\phi_2(r,a)) 
&= \what(x, r+c^*+\phi_2(r,\a)).
\end{aligned} \end{equation}

We are now prepared to prove the theorem using an inductive argument.
Let $A$ be the level set defined in the statement of the theorem.
Without loss of generality, we write $A=[-r_0,r_0]$, for some $r_0>0$.
We have
$$
g(u_0(x)) = \begin{cases}
1, & -r_0 \leq x \leq r_0, \\
0, & \text{else}.
\end{cases}
$$
for all $x$. Applying the convolution operator yields
$$ \begin{aligned}
u_1(x) &= w_1(x-r_0) - w_1(x+r_0) \\
&= \what(x-r_0,r_0).
\end{aligned}
$$

Let $(r_n)_{n=0}^{\infty}$ be a sequence of real numbers satisfying
\begin{equation} \label{recurrence}
r_{n+1} = r_n + c^* + \phi_2(r_n,a).
\end{equation}
If $r_0$ is sufficiently large, we can show that $r_n$ is strictly increasing and unbounded.
This is due to our assumption that $c^*>0$.
Since $\phi_2(r,a)\to c^*$, it suffices to assume $|\phi_2(r,a)-c^*|<c^*$, i.e. $r_0$ is sufficiently large so that $\phi_2(r,a)>0$ for all $r>r_0$.
Plugging this into the recurrence \eqref{recurrence} yields $r_{n+1}-r_n = c^* + \phi_2(r,a) > c^* > 0$.

This argument shows that, assuming $r_0$ is sufficiently large, equation \eqref{appliedqtwice} may be applied recursively.
From the principle of induction, it follows that $u_{2n+1}(x) = \what(x-r_n,r_n)$ for all $n \geq 0$.
Knowing that $(r_n)$ is unbounded allows us to use equation \eqref{recurrence} to compute the asymptotic spreading speed as the following limit.
$$
\lim_{n\to\infty} (r_{n+1}-r_n) = \lim_{n\to\infty} (c^* + \phi_2(r_n,a)) = c^* + \lim_{r\to\infty} \phi_2(r,a) = 2c^*.
$$
This completes the proof.
\end{proof}

\section{Examples}
In this section, we construct the periodic traveling wave solution for uniform and Laplace kernels and derive formulas for the mean spreading speed $c^*$ in terms of the model parameters.

\subsection{Laplace kernel} 
We applied our main results to the Laplace dispersal kernel, defined
\begin{equation} \label{laplacekernel}
k(x) = \frac{1}{2} e^{-|x|}.
\end{equation}

$k(x)$ is symmetric and has connected support, satisfying \hypref{h1}, \hypref{h2}, and \hypref{h3}.
We will now prove that $k(x)$ satisfies \hypref{h5}.
Let $A,B\in\mathbb R$ with $A>B$, let $\lambda \in (0,1)$, and define $f(x) = k(x-A) - \lambda k(x-B)$.
For $r>0$ sufficiently large, we have
$$
f(x-r) - f(-x-r) = k(x-r-A) - \lambda k(x-r-B) - k(x+r+A) + \lambda k(x+r+B)
$$
$$
= \begin{cases}
\frac{1}{2} e^{x}\left(-e^{r+A}+me^{r+B}-me^{-r-B}+e^{-r-A}\right) , & x < -r-A, \\
\frac{1}{2}e^{-x-r-A}+\frac{1}{2}e^{x}\left(e^{-r-A}-me^{-r-B}+me^{r+B}\right) , & -r-A \leq x \leq -r-B, \\
\frac{1}{2}e^{-x}\left(-e^{-r-A}+me^{-r-B}\right)+\frac{1}{2}e^{x}\left(-me^{-r-B}+e^{-r-A}\right) , & -r-B < x < r+B, \\
\frac{1}{2}e^{-x}\left(-e^{-r-A}+me^{-r-B}-me^{r+B}\right)+\frac{1}{2}e^{x-r-A}, & r+B \leq x \leq r+A \\
e^{-x}\left(-e^{-r-A}+me^{-r-B}-me^{r+B}+e^{r+A}\right), & x > r+A.
\end{cases}
$$

$f$ is monotone on $(-\infty,-r-A)\cup(r+A,\infty)$ with $f(\pm\infty)=0$; thus it cannot have any sign changes there.
On each of the intervals $[-r-A,-r-B]$, $(-r-B,r+B)$, and $[r+B,r+A]$, $f$ has the form $Ce^x+De^{-x}$.
Depending on the sign of $C$ and $D$, $f$ can have either one or zero sign changes on each interval.
Summing over each interval, it follows that $f$ has at most $3$ sign changes.
This proves \hypref{h5}; the proof of \hypref{h4} follows a similar argument.

Assuming $||w_2||_\infty<b$, the periodic traveling wave profiles are given by
\begin{equation} \label{w1laplace}
w_1(x) =  \begin{cases} 
1 - \frac{1}{2} e^{x}, & \text{if } x \leq 0, \\
\frac{1}{2} e^{-x}, & \text{if } x > 0,
\end{cases}
\end{equation}
and
\begin{equation} \label{w2laplace}
w_2(x-c^*) = \begin{cases}
\m + \left( \frac{1}{2}(\mtwo-\m)e^{-\Phi_1(\b)} - \frac{1}{2} e^{-\Phi_1(\a)} \right) e^x, & x < \Phi_1(\b), \\
\mtwo - \frac{1}{2} e^{-\Phi_1(\a)} e^x - \frac{1}{2} (\mtwo-\m) e^{\Phi_1(\b)}  e^{-x}, & \Phi_1(\b) \leq x \leq \Phi_1(\a), \\
\left( \frac{1}{2}e^{\Phi_1(\a)} - \frac{1}{2}(\mtwo-\m) e^{\Phi_1(\b)} \right) e^{-x}, & \Phi_1(\a) < x, 
\end{cases} \end{equation}
with
\begin{equation} \label{laplacecstar}
c^* = \frac{1}{2} \log \begin{cases}
\frac{1}{2\a} \left( e^{\Phi_1(\a)} -(\mtwo-\m)e^{\Phi_1(\b)}\right), & \a \leq \ell_a, \\
 e^{\Phi_1(\a)} \left( \mtwo - \a + \sqrt{(\mtwo-\a)^2-(\mtwo-\m)e^{\Phi_1(\b)-\Phi_1(\a)}}\right), & \ell_a < \a < \ell_b, \\
2\,(\m-\a)\left[e^{-\Phi_1(\a)}-(\mtwo-\m)e^{-\Phi_1(\b)}\right]^{-1}, & \a \geq \ell_b.
\end{cases}
\end{equation}
$\ell_a$ and $\ell_b$ are constants given by
\begin{equation}
\ell_a = \frac{1}{2}\left(1 - (\mtwo-\m) e^{\Phi_1(\b)-\Phi_1(\a)} \right), \quad 
\ell_b = \frac{1}{2} \left( 1+\m- e^{\Phi_1(\b)-\Phi_1(\a)} \right).
\end{equation}
The formulas for $\Phi_1$ and $\Phi_2$ are given by
\begin{equation}
\Phi_1(\ell) = \begin{cases} -\log(2\ell), &\ell\leq \frac{1}{2}, \\ \log(2-2\ell), & \ell > \frac{1}{2}, \end{cases}
\end{equation}
and
\begin{equation}
\Phi_2(\ell) = -c^* + \log \begin{cases}
\frac{1}{2\ell}\left(e^{\Phi_1(\a)} -(\mtwo-\m)e^{\Phi_1(\b)} \right), & \ell \leq w_2(\Phi_1(\a)), \\
e^{\Phi_1(\a)} \left( \mtwo - \ell+ \sqrt{(\mtwo-\ell)^2-(\mtwo-\m)e^{\Phi_1(\b)-\Phi_1(\a)}}\right), & w_2(\Phi_1(\a)) < \ell < w_2(\Phi_1(\b)), \\
2\,(\m-\ell) \left[e^{-\Phi_1(\a)}-(\mtwo-\m)e^{-\Phi_1(\b)}\right]^{-1}, & \ell\geq w_2(\Phi_1(\b)).
\end{cases}
\end{equation}

The critical spreading speed $c^*$ for the Laplace kernel thus falls into one of three cases:
\begin{enumerate}[{Case} 1.]

\item $\a\leq\frac{1}{2}$, $\b\leq\frac{1}{2}$:
$$ c^* = \frac{1}{2}\log \begin{cases}
\frac{\b-\a(1-\m)}{4\a^2\b}, & \a < \frac{\b-\a(1-\m)}{2\b}, \\
\frac{1-\a + \sqrt{\b(1-\a)^2 - \a(1-\m)}}{2\a\sqrt\b}, & \frac{\b-\a(1-\m)}{2\b} \leq \a \leq \frac{\b(1+\m)-\a}{2\b}, \\
\frac{\a-\m}{\b(1-\m)-\a}, & \a > \frac{\b(1+\m)-\a}{2\b}.
\end{cases} $$

\item $\a\leq\frac{1}{2}$, $\b>\frac{1}{2}$:
$$ c^* = \frac12 \log \begin{cases}
\frac{1-4\a(1-\m)(1-\b)}{4\a^2}, & a<\frac{1-4\a(1-\m)(1-\b)}{2}, \\
\frac{1-\a + \sqrt{(1-\a)^2 - 4\a(1-\m)(1-\b)}}{2\a}, & \frac{1-4\a(1-\m)(1-\b)}{2}\leq\a\leq\frac{1+\m-4\a(1-\b)}{2}, \\
\frac{4(\a-\m)(1-\b)}{1-\m-4\a(1-\b)}, & \a>\frac{1+\m-4\a(1-\b)}{2}.
\end{cases} $$

\item $\a>\frac{1}{2}$, $\b>\frac{1}{2}$:
$$ c^* = \frac12 \log \begin{cases}
\frac{1-\a-(1-\m)(1-\b)}{\a}, & \a<\frac{1-\a-(1-\m)(1-\b)}{2(1-\a)}, \\
2\,(1-\a)\left(1-\a + \sqrt{\frac{(1-\a)^3-(1-\m)(1-\b) }{1-\a}}\right), & \frac{1-\a-(1-\m)(1-\b)}{2(1-\a)}\leq\a\leq\frac{\b-1+\m(1-\a)}{2(1-\a)}, \\
\frac{4(\m-\a)(1-\a)(1-\b)}{1-\b+\m(1-\a)}, & \a>\frac{\b-1+\m(1-\a)}{2(1-\a)}.
\end{cases} $$
\end{enumerate}

Figure \ref{fig:heatmap} shows the results of a numerical simulation with growth parameters were $a=0.3$, $m=0.5$, and $b=0.8$ with the Laplace dispersal kernel, for an average spreading speed of $c^* = \frac{1}{2} \ln \frac{22}{9}$.
Figures \ref{fig:subim1} and \ref{fig:subim2} show heatmaps of the periodic traveling wave solution and a spreading solution with compactly supported initial condition.
Figure \ref{fig:laplacewave} shows a numerical simulation of the periodic traveling wave solution, while Figure \ref{fig:laplacecompact2} shows the asymptotic convergence of the intermediate wave-speed sequence $c_n(\ell)$, evaluated at the Allee threshold $\ell = \a$.

\begin{figure}[h]
\begin{subfigure}{0.5\textwidth}
\includegraphics[width=\linewidth]{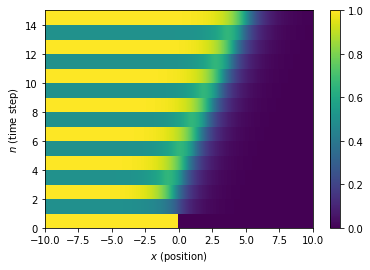}
\caption{Heatmap of the periodic traveling wave.}
\label{fig:subim1}
\end{subfigure}
\begin{subfigure}{0.5\textwidth}
\includegraphics[width=\linewidth]{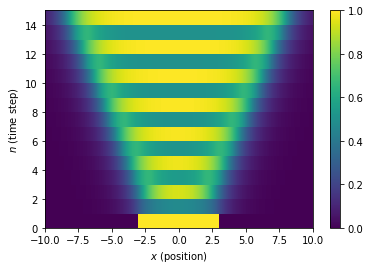}
\caption{Heatmap of a periodic spreading solution.}
\label{fig:subim2}
\end{subfigure}
\begin{subfigure}{0.5\textwidth}
 \includegraphics[width=\linewidth]{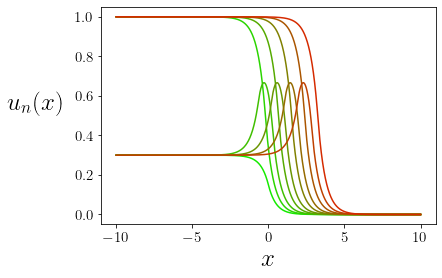}
\caption{Population density curve of the periodic traveling wave.}
\label{fig:laplacewave}
\end{subfigure}
\begin{subfigure}{0.5\textwidth}
 \includegraphics[width=\linewidth]{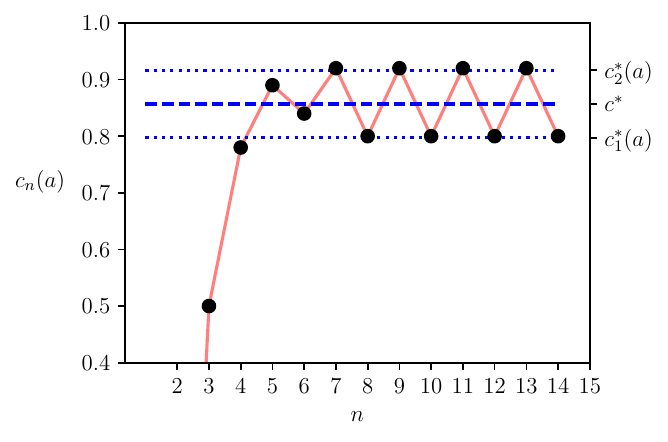}
 \caption{Time series plot of intermediate spreading speed.}
 \label{fig:laplacecompact2}
\end{subfigure}
\caption{Heatmaps of (a) the periodic traveling wave, and (b) the corresponding spreading solution with compactly supported initial data, and Laplace dispersal kernel.
In figure (d), the horizontal blue lines correspond to the theoretical values derived in Corollary \ref{cor1}.}
\label{fig:heatmap}
\end{figure}

\subsection{Uniform kernel}
The uniform dispersal kernel is given by
\begin{equation}
k(x) = \begin{cases}
\frac{1}{2}, & \text{if } |x|\leq 1, \\
0, & \text{if } |x| > 1.
\end{cases} \end{equation}
Like the previous example, $k(x)$ is symmetric around zero.
We can verify \hypref{h4} by letting $A>B$ and $\lambda \in (0,1)$ and defining $f(x)=k(x-A)-\lambda k(x-B)$.
There are two cases:
\begin{enumerate}[{Case} 1.]
\item If $|A-B|<2$, then
$$
f(x) = \begin{cases}
0, & x \leq B-1, \\
-\frac{\lambda}{2}, & B-1 < x \leq A-1, \\
\frac{1-\lambda}{2}, & A-1 < x < B+1, \\
\frac{1}{2}, & B+1 \leq x < A+1 \\
0, & x \geq A+1.
\end{cases}
$$

\item If $|A-B| \geq 2$, then
$$
f(x) = \begin{cases}
0, & x \leq B-1, \\
-\frac{\lambda}{2}, & B-1 < x < B+1, \\
0, & B+1 \leq x \leq A-1, \\
\frac{1}{2}, & A-1 < x < A+1 \\
0, & x \geq A+1.
\end{cases}
$$
\end{enumerate}
In both cases, the number of sign changes is exactly 1.
To extend this argument to \hypref{h5}, observe that since $k(x)$ has compact support, we have
$$
\begin{aligned}
f(x-r) - f(r-x) &= k(x-r-A) - \lambda k(x-r-B) - k(x+r+A) + \lambda k(x+r+B) \\
&= \begin{cases}
-f(r-x), & x < -r+1, \\
0, & -r+1 \leq x \leq r-1, \\
f(x-r), & x > r-1,
\end{cases}
\end{aligned}
$$
for sufficiently large $r$.

The alternating wave profiles are given by
\begin{equation} \label{uniformw1}
\begin{aligned}
w_1(x) 
= \begin{cases}
1, & x \in (-\infty, -1), \\
\frac{1}{2}-\frac{1}{2}x, & x \in [-1, 1] ,\\
0, & x \in (1, \infty),
\end{cases}
\end{aligned} \end{equation}
and 
\begin{equation} \label{uniformw2}
\begin{aligned}
w_2(x+c^*)
= \begin{cases}
m,
& x \in (-\infty, -2\b), \\
\frac{1-m}{2}x + m + b - mb,
& x \in [-2\b, 
 -2\a), \\
-\frac{m}{2}x +m+b- mb-a,
& x \in [-2\a, 2-2\b), \\
-\frac{1}{2} x-a+1,
& x \in [2-2\b, 2-2\a], \\
0,
& x \in (2+2\a,\infty),
\end{cases}
\end{aligned} \end{equation}
with
\begin{equation} \label{clinear}
c^* = \begin{cases}
1 - 2a, & \text{if } a \leq b/2, \\
1 -b + \frac{b - 2a}{m}, & \text{if }a > b/2.
\end{cases}
\end{equation}

$w_2$ has a global maximum at $x=-2\a$ so that $||w_2||_\infty=m+(b-a)(1-m)$.
Thus, a sufficient condition for the existence of the periodic traveling wave is $w_2(-2\a)=m+(b-a)(1-m) < \b$.

The intermediate wave speeds are given by
\begin{equation} \begin{aligned}
c_1^*(\ell) 
&= 1-2\ell
\end{aligned}
\end{equation}
and
\begin{equation} \begin{aligned}
c_2^*(\ell) 
&= \begin{cases}
1 - 4a + 2\ell, & a \leq b/2, \\
1 - 2b + \frac{2b-4a}{m} + 2\ell, & a > b/2 ,
\end{cases}
\end{aligned} \end{equation}
for $0 < \ell < m + (b-a)(1-m)$.
By evaluating the intermediate wave speeds at $\ell=\a$, we obtain $c_1^*(\a)=c_2^*(\a)$ if $\a \leq \b/2$, and $|c_1^*(\a)-c_2^*(\a)|=2(2\a-\b)(\frac{1-\m}{\m})>0$ if $a>b/2$.
So for $a>b/2$, the traveling wave is periodic with two different intermediate wave speeds.
Furthermore, the difference between these two intermediate speeds is increasing with respect to $a$, the Allee threshold, and decreasing with respect to $\b$, the overcompensation threshold.

Figure \ref{fig:subim3} depicts the periodic traveling wave solution \eqref{ptw} with the uniform kernel with growth parameters are $\a=0.325$, $\b=0.6$, and $\m=0.45$.
The mean spreading speed can be analytically calculated as $c^*=\frac{13}{45}$.
The intermediate speeds are $c_1^*(\a)=\frac{7}{20}$ and $c_2^*(\a)=\frac{41}{180}$.
Figure \ref{fig:subim4} demonstrates the spreading phenomena with initial domain size $r_0=1$.
Since our uniform kernel is compact with support $[-1,1]$, a sufficient lower bound on the initial domain size is $r_0 \geq 3$.

\begin{figure}[h]
\begin{subfigure}{0.5\textwidth}
\includegraphics[width=\linewidth]{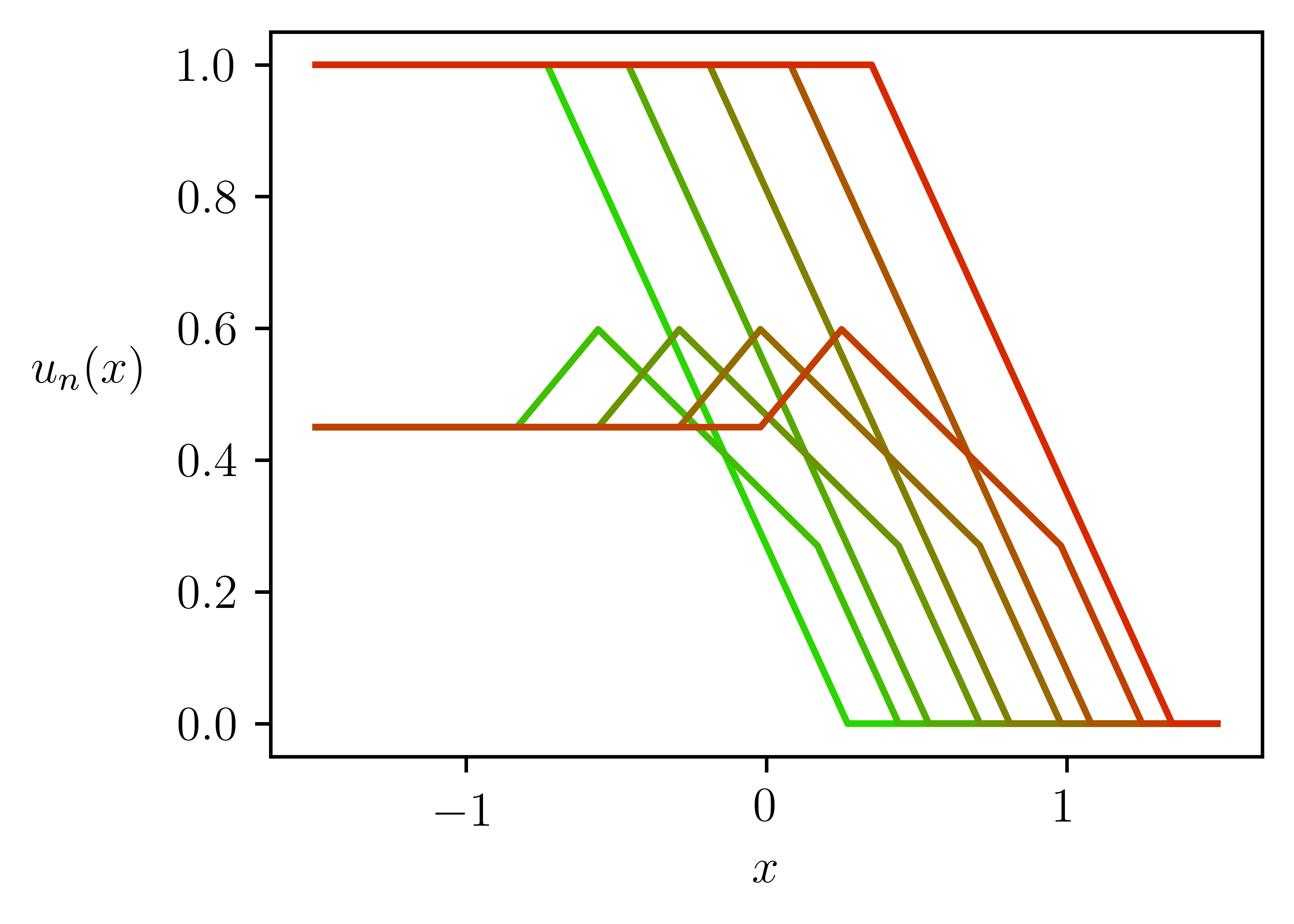}
\caption{Heaviside initial data.}
\label{fig:subim3}
\end{subfigure}
\begin{subfigure}{0.5\textwidth}
\includegraphics[width=\linewidth]{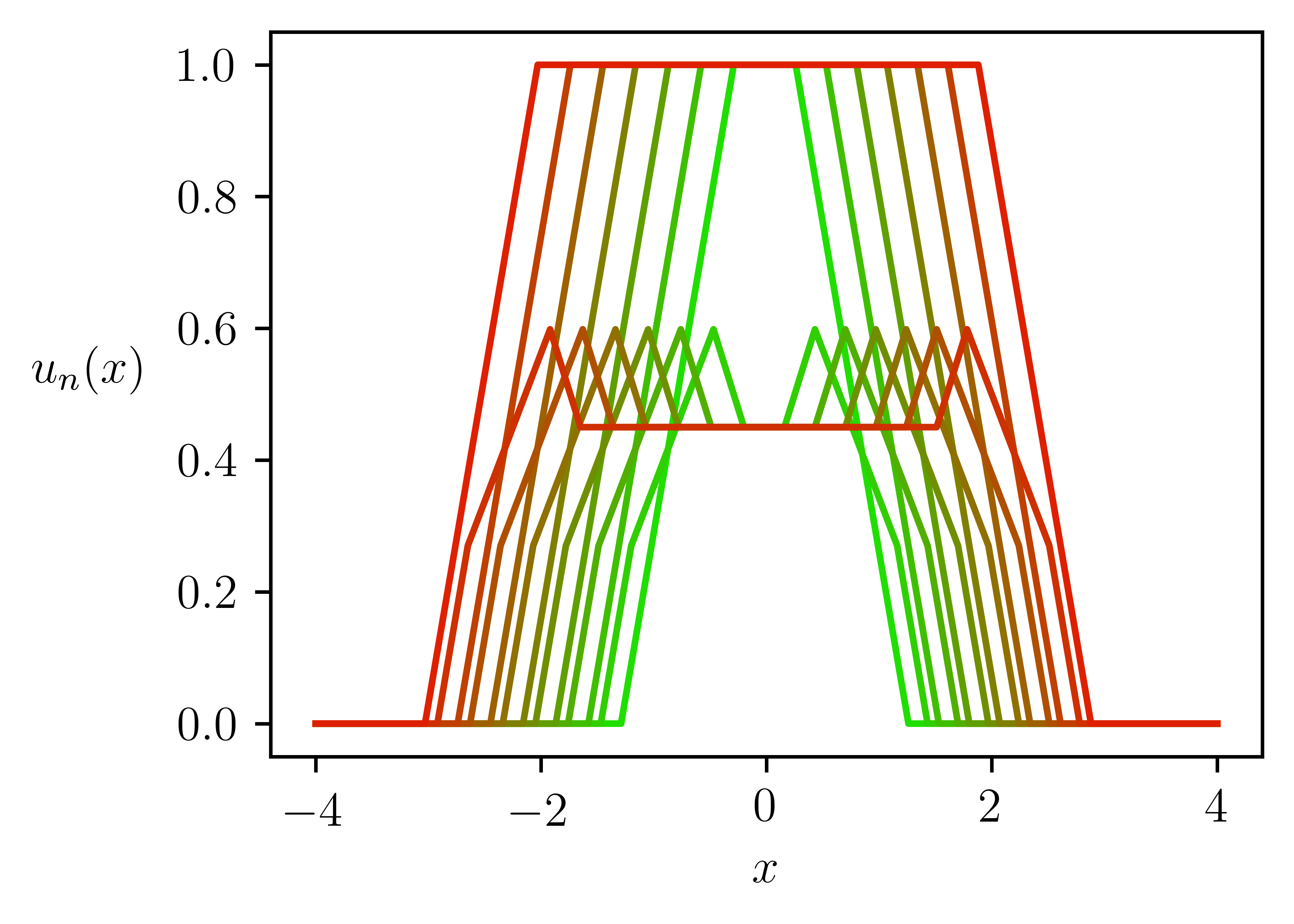}
\caption{Compact initial data ($r_0=1$).}
\label{fig:subim4}
\end{subfigure}
\caption{Population density plots of (a) the periodic traveling wave, and (b) a spreading solution with uniform dispersal.}
\label{fig:uniformplot}
\end{figure}


\begin{thebibliography}{9}
\bibitem{all} W. C. Allee. 1931.
Animal Aggregations. A Study on General Sociology.
University of Chicago Press, Chicago, IL.

\bibitem{hh} A. Hastings and K. Higgins. 1994.
Persistence of transients in spatially structured ecological models.
Science {\bf 263}: 1133-1136.

\bibitem{ks} M. Kot and W. M. Schaffer. 1986.
Discrete-time growth-dispersal models.
Math. Biosci. {\bf 80}: 109-136.

\bibitem{kot89} M. Kot. 1989.
Diffusion-driven period doubling bifurcations.
Biosystems {\bf 22}: 279-287.

\bibitem{kot92} M. Kot. 1992.
Discrete-time traveling waves: Ecological examples.
J. Math. Biol. {\bf 30}: 413-436.

\bibitem{kot1} M. Kot, M. A. Lewis, and P. van den Driessche. 1996.
Dispersal data and the spread of invading organisms.
Ecology {\bf 77}: 2027-2042.

\bibitem{kotbook} M. Kot. 2001.
Elements of Mathematical Ecology.
Cambridge University Press. Cambridge, United Kingdom.

\bibitem{lui82a} R. Lui. 1982.
A nonlinear integral operator arising from a model in population genetics. I. Monotone initial data.
SIAM. J. Math. Anal. {\bf 13}: 913-937.

\bibitem{lui82b} R. Lui. 1982.
A nonlinear integral operator arising from a model in population genetics. II. Initial data with compact support.
SIAM. J. Math. Anal. {\bf 13}: 938-953.

\bibitem{lui83} R. Lui. 1983.
Existence and stability of traveling wave solutions of a nonlinear integral operator.
J. Math. Biol. {\bf 16}:199-220.

\bibitem{lut} F. Lutscher. 2019.
Integrodifference Equations in Spatial Ecology.
Springer.

\bibitem{nkl} M. Neubert, M. Kot, and M. A. Lewis. 1995.
Dispersal and pattern formation in a discrete-time predator-prey model.
Theor. Pop. Biol. {\bf 48}: 7-43.

\bibitem{otto} G. Otto. 2017.
Non-spreading Solutiona in a Integro-Difference Model Incorporating Allee and Overcompensation Effects.
Ph. D thesis, University of Louisville.

\bibitem{slatkin} M. Slatkin. 1973.
Gene flow and selection in a cline.
Genetice {\bf 75}: 733-756.

\bibitem{pnas} L. L. Sullivan, B. Li, T. E. X. Miller, M. G. Neubert, and A. K. Shaw. 2017.
Density dependence in demography and dispersal generates fluctuating invasion speeds. Proc. Natl. Acad. Sci. USA {\bf
114}: 5053-5058.

\bibitem{wang} M. H. Wang, M. Kot, and M. G. Neubert. 2002.
Integrodifference equations, Allee effects, and invasions.
J. Math. Biol. {\bf 44}: 150-168.

\bibitem{w78} H. F. Weinberger. 1978.
Asymptotic behavior of a model in population genetics.
Nonlinear Partial Differential Equations and Applications, ed. J. M. Chadam. Lecture Notes in Mathematics {\bf 648}: 47-96. Springer-Verlag, Berlin.

\bibitem{wein82} H. F. Weinberger. 1982.
Long-time beahvior of a class of biological models.
SIAM. J. Math. Anal. {\bf 13}: 353-396.
\end{thebibliography}
\end{document}